\documentclass[conference]{IEEEtran}
\IEEEoverridecommandlockouts

\usepackage{cite}

\ifCLASSINFOpdf
   \usepackage[pdftex]{graphicx}
   \DeclareGraphicsExtensions{.pdf,.jpeg,.png}
\else
  \usepackage[dvips]{graphicx}
   \DeclareGraphicsExtensions{.eps}
\fi
\usepackage[cmex10]{amsmath}
\usepackage{amssymb}
\usepackage{amsthm}

\newtheorem{theo}{Theorem}
\newtheorem{lem}{Lemma}

\usepackage{algorithmic}

\usepackage{balance}
\usepackage{array}
\usepackage{fixltx2e}
\usepackage{url}

\begin{document}
%
\title{An Achievable Rate for an Optical Channel with Finite Memory}

\author{\IEEEauthorblockN{K Gautam Shenoy, Vinod Sharma}
\IEEEauthorblockA{Dept. of ECE, Indian Institute of Science, Bangalore-560012\\
Email: konchady,vinod@ece.iisc.ernet.in}}


%


\maketitle

\begin{abstract}
A fiber optic channel is modeled in a variety of ways; from the simple additive white complex Gaussian noise model, to models that incorporate memory in the channel. Because of Kerr nonlinearity, a simple model is not a good approximation to an optical fiber. Hence we study a fiber optic channel with finite memory and provide an achievable bound on channel capacity that improves upon a previously known bound.
\end{abstract}


%
\IEEEpeerreviewmaketitle

\section{Introduction}
The study of the fiber optic channel (FOC) from an information theoretic standpoint has been an interesting subject in recent times. In particular, estimating the channel capacity, has been a very daunting task owing to several factors. Firstly, from \cite{Duris2,Kram2,Ess1}, there is no one agreed upon model which incorporates all the non-linear effects that occur in an optical fiber. Thus we have several models, each limited to an aspect under study. Secondly, for most of these models, the channel capacity is theoretically unknown and, while there may be good lower bounds available for these channels, only a few of them have useful upper bounds on channel capacity. 

Nonlinearities of a FOC from an information theoretic standpoint were first studied by Stark and Mitra in \cite{Strk1,Strk2}. Lower and upper bounds on channel capacity for free space optical intensity channels were obtained in \cite{lapi} under peak and average power constraints. Input dependent Gaussian noise channels were studied in \cite{Mos1} where tight upper and lower bounds on channel capacity were derived under a variety of power constraints. The conditions for these bounds were further refined in \cite{CDN}. In \cite{dar}, a lower bound on the capacity of an additive white complex Gaussian noise (AWCGN) channel with input perturbed by phase noise was derived. Also, an upper bound on the cascade of nonlinear and noisy channels was derived in \cite{Kram2}. A survey of various types of channels for optical fibers is carried out in \cite{Kram1}.
A more general fiber optical channel with memory (also known as the finite memory Gaussian noise model) was studied in \cite{Duris1}. The strategy used was that of block i.i.d. coding to mitigate the effects of memory. 

A trait of most, if not all, of the models is that the channel is modeled as an AWCGN channel with the necessary modifications. The effects of Kerr nonlinearity introduce memory in the channel and it is essential that it be taken into consideration. However, the tradeoff here is that introducing memory makes the channel harder to analyze and, for this channel, the capacity is unknown. 

In this paper, we study the model of \cite{Duris1} and provide a closed-form, improved lower bound to the capacity of a FOC with finite memory. The technique used is that of comparison with a suitable auxiliary channel \cite{arim,blah}. Furthermore, we also provide an algorithm derived from\cite{vs1,arn} to obtain numerical lower bounds.
 
The paper is organized as follows. We provide the notation and channel model in Section \ref{Prem}. A  lower bound on the capacity for the FOC with a finite memory is stated in Section \ref{Lb} and proved in Section \ref{proo}. Section \ref{algorz} describes an algorithm to compute a better achievable rate. Section \ref{numer} compares the lower bound derived in Section \ref{proo}, the rate computed via the algorithm in Section \ref{algorz} and the lower bound available in \cite{Duris1}. Section \ref{concl} concludes the paper.


\section{Preliminaries} \label{Prem}
\subsection{Notations}
We use lower case letters (e.g., $x$) to refer to scalars and upper case letters (e.g., X) for random variables. Vectors are represented using boldface letters (e.g., $\mathbf{x}$) and when it is required to indicate the length of a vector, we shall denote it by $\mathbf{x}^n :=(x_1,x_2,\cdots,x_n)$. Let $\mathcal{X}$ denote the input alphabet and $\mathcal{Y}$ the output alphabet of the channel. We denote the set of real numbers by $\mathbb{R}$ and complex numbers by $\mathbb{C}$.

A discrete time memoryless channel is represented by a random transformation $W(y|x),~ x\in \mathcal{X}, y \in \mathcal{Y}$, which for continuous alphabet, is the channel conditional probability density function. Given an input distribution $P$, we shall denote the corresponding output distribution by 
\begin{equation}
PW(y) := \int_{x \in \mathcal{X}} P(x)W(y|x) dx.
\label{outp}
\end{equation}

The differential entropy of a continuous random variable $X$ is denoted by $h(X)$ and the mutual information between random variables $X$ and $Y$ will be denoted by $I(X;Y)$. If the distributions need to be emphasized, we will use the equivalent notation $I(P_X;W)$ for mutual information, where $P_X$ is the input distribution or density and the channel is $W$. These notations are taken from \cite{cover,csis}. For the memoryless channel, the capacity
\begin{equation}
C= \sup_{P_X} I(P_X;W),
\label{Capori}
\end{equation}
where $P_X$ is the set of distributions that satisfy the needed constraints (e.g., average power constraints).

For any channel with finite memory, assuming it is information stable (see \cite{verdu}), the channel capacity
\begin{equation}
C = \lim_{n \to \infty} \sup_{P_{\mathbf{X}^n}} \frac{I(\mathbf{X}^n;\mathbf{Y}^n)}{n},
\label{chancap}
\end{equation}
where the supremum is taken over all input distributions that satisfy the power constraint. Now from the properties of mutual information, for any positive integer $N$,
\begin{equation}
\frac{I(\mathbf{X}^{n+N};\mathbf{Y}^{n+N})}{n+N}.\frac{n+N}{n} \geq \frac{I(\mathbf{X}^{n+N};\mathbf{Y}^n)}{n} \geq \frac{I(\mathbf{X}^{n};\mathbf{Y}^n)}{n}.
\label{cap3z}
\end{equation}
Thus we get capacity 
\begin{equation}
C = \lim_{n\to\infty} \sup_{P_{\mathbf{X}^{n+N}}}\frac{I(\mathbf{X}^{n+N};\mathbf{Y}^n)}{n}.
\label{capf}
\end{equation}
This is the form we will use in this paper.
\subsection{Kerr Nonlinearity in FOCs}
Wave propagation in an optical fiber is described by the non-linear Schr\"{o}dinger wave equation (see \cite{Agar,Ess1})
\begin{equation}
\frac{\partial x}{\partial z} + \frac{j\beta}{2}\frac{\partial^2 x }{\partial t^2} -j\gamma|x|^2x = 0,
\label{swe}
\end{equation}
where $x=x(t,z)$ is the complex envelope of the optical signal at time $t$ and distance $z$, $\beta$ represents the group velocity dispersion and $\gamma$ is the Kerr nonlinearity coefficient. While designing the channel model, Kerr non-linearity introduces memory in the channel and manifests itself as a third power of the superposition of current input as well as past and future inputs. Hence the channel model is non-causal.

Another phenomenon that occurs is that of Amplified Spontaneous Emission (ASE) in the optical fiber (see \cite{Ess1}). This is modeled as an additive noise (see Section \ref{chanmod}) in the channel.

\subsection{Channel Model}\label{chanmod}
The channel model we consider was proposed and studied in \cite{Duris1}. The model (see Fig. \ref{fig1}), after a suitable modification (see the remark below), is given by 
\begin{equation}
Y_i = X_i + A_i + Z_i\sqrt{\eta S_i^3},
\label{chan1}
\end{equation}
for $1\leq i\leq n$ where 
\begin{enumerate}
	\item $X_i \in \mathbb{C}$ is the input to the channel.
	\item $Y_i \in \mathbb{C}$ is the channel output.
	\item $A_i \in \mathbb{C}$ is the noise due to ASE modeled as circularly symmetric complex Gaussian random variable with 
	zero mean and variance $\sigma_A^2$. We assume $\{A_i\}$ are independent and identically distributed (i.i.d.).
	\item $Z_i \in \mathbb{C}$ is a circularly symmetric complex normal random variable with zero mean and unit variance. We assume $\{Z_i\}$ are i.i.d.
	\item $S_i \geq 0$ models the memory in the channel and is given by
	\begin{equation}
	S_i = \frac{1}{2N+1}\sum_{k=i-N}^{i+N} |X_k|^2,
	\label{Si}
	\end{equation}
	where $N \geq 0$ is the memory in the channel. When $i\leq N$, some of the indices $k$ in the sum may be negative (we take those corresponding $|X_k|^2$ as 0). However this does not affect our capacity analysis and so we may assume henceforth that $i>N$. Note that $\{S_i\}$ are not independent for $N\geq 1$. $\eta$ is a constant that depends on the non-linearity parameter $\gamma$ and the design of the optical fiber. Both $S_i$ and $\eta$ together model the Kerr nonlinearity.
	\item $\{X_i\}$, $\{A_i\}$ and $\{Z_i\}$ are mutually independent.
	\item We assume that the channel has an average input power constraint,
	\begin{equation}
	\frac{1}{n}\sum_{i=1}^n|X_i|^2 \leq P \quad a.s.
	\label{pcon}
	\end{equation}
\end{enumerate}
where $n$ is the codeword length.
\begin{figure}%
\centering
\includegraphics[scale=0.7]{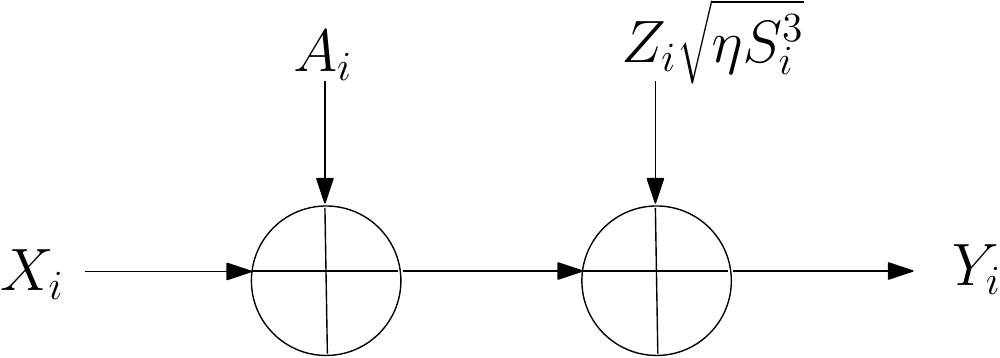}%
\caption{The finite memory optical channel.}%
\label{fig1}%
\end{figure}

We remark that this channel is functionally equivalent (i.e., has the same mutual information and capacity) as the following channel, 
\begin{equation}
Y_i = X_i +  Z_i\sqrt{\sigma_A^2+\eta S_i^3},
\label{orich}
\end{equation}
which was originally studied in \cite{Duris1}. This equivalence follows from the following lemma.
\begin{lem}
Consider two channels
\begin{enumerate}
	\item $Y_i = X_i + Z_i\sqrt{\sigma_A^2 + U_i}$, \label{c1}
	\item $T_i = X_i + A_i\sigma_A +  V_i\sqrt{ U_i}$, \label{c2}
\end{enumerate}
where $X_i$ is the input to the channel, $Z_i$, $A_i$ and $V_i$ are i.i.d. circularly symmetric complex Gaussian random variables with zero mean and unit variance and are mutually independent of each other. Let $U_i$ be a non negative random variable which is  independent of $Z_i$, $A_i$ and $V_i$ but is dependent on $\{X_k\}_{k=i-N}^{i+N}$. Then it follows that for any $n\geq 1$,
\begin{equation}
I(\mathbf{X}^{n+N} ; \mathbf{Y}^{n}) = I(\mathbf{X}^{n+N} ; \mathbf{T}^{n}).
\label{inf1}
\end{equation}
Moreover, this implies that the two channels have the same capacity.
\end{lem}
\begin{proof}
We shall first prove the result for real valued random variables, replacing circularly symmetric Gaussian with standard Gaussian. To show that the mutual informations are equal, it suffices to show that 
\begin{equation}
(\mathbf{X}^{n+N},\mathbf{Y}^{n}) \stackrel{d}{=} (\mathbf{X}^{n+N} , \mathbf{T}^{n}),
\label{diste}
\end{equation}
where $\stackrel{d}{=}$ denotes equal in distribution.
As we are working with vectors, we will introduce some useful shorthand notation. We shall write 
\begin{IEEEeqnarray}{rCl}
\{[Y]\leq [y]\} &=& \{Y_1 \leq y_1, Y_2\leq y_2 ,\cdots, Y_n \leq y_n\} \\
\{[X_N] \leq [x_N]\} &=& \{X_1 \leq x_1,\cdots, X_{n+N} \leq x_{n+N}\}
\label{shor}
\end{IEEEeqnarray}

Hence we have
\begin{IEEEeqnarray}{rCl}
&&Pr([X_N] \leq [x_N],[Y]\leq [y]) \\
&=& \int_{-\infty}^{x_1}\int_{-\infty}^{x_2}\cdots\int_{-\infty}^{x_{n+N}}f_{[X_N]}([x_N])\notag \\
&&\times Pr([Y]\leq [y]|[X_N] = [x_N])d[y]d[x_N].
\label{equ2}
\end{IEEEeqnarray}
Now consider $Pr([Y]\leq [y]|[X_N] = [x_N])$. This is simplified as 
\begin{IEEEeqnarray}{rCl}
&&Pr([Y]\leq [y]|[X_N] = [x_N]) \\
&=&Pr([X+Z\sqrt{\sigma_A^2 + U}]\leq [y]|[X_N] = [x_N]) \\
&=&Pr([Z\sqrt{\sigma_A^2 +u}]\leq [y-x]|[X_N] = [x_N]) \\
&=&Pr([Z\sqrt{\sigma_A^2 +u}]\leq [y-x]),
\label{eq3}
\end{IEEEeqnarray}
where given $[X_N]=[x_N]$, $[U]$ becomes a constant vector $[u]$ and in the last step, we have used the independence of $[Z]$. 

Now we have
\begin{IEEEeqnarray}{rCl}
 &&Pr\left([Z\sqrt{\sigma_A^2 + u}]\leq [y-x]\right)\\
&=& \prod_{i=1}^n Pr\left(Z_i\sqrt{\sigma_A^2 + u_i}\leq y_i-x_i\right)\\
&=&\prod_{i=1}^n Pr\left(A_i\sigma_A + V_i\sqrt{ u_i}\leq y_i-x_i\right)\label{equ3}\\
&=&Pr([A\sigma_A +V\sqrt{ u}]\leq [y-x]),
\label{equ4}
\end{IEEEeqnarray}
where in (\ref{equ3}), we used the fact that two Gaussian random variables have the same distribution if they have the same mean and variance. Substituting this in (\ref{equ2}), and essentially reversing the steps gives us the required result.

To extend for complex random variables, we need to further split each complex Gaussian variable into its real and imaginary part and use that a circularly symmetric complex Gaussian is obtained by the joint density of two independent zero mean real valued Gaussian random variables with each having half the total variance. 

Putting it all together, we get 
\begin{equation}
I(\mathbf{X}^{n+N} ; \mathbf{Y}^{n}) = I(\mathbf{X}^{n+N} ; \mathbf{T}^{n}),
\label{inf2}
\end{equation}
for any joint valid input distribution $f_{\mathbf{X}^{n+N}}$. 

From (\ref{capf}) and (\ref{inf2}) we get the desired results.
\end{proof}
\section{Lower Bound on Channel Capacity}\label{Lb}
Since we are concerned with achievability, we can lower bound capacity as
\begin{equation}
C \geq \lim_{n \to \infty} \frac{I(\mathbf{X}^{n+N};\mathbf{Y}^n)}{n},
\label{capb1}
\end{equation}
for any input distribution $P_{\mathbf{X}^{n+N}}$ that satisfies the power constraint.

Given $\delta > 0$ small enough so that $P-\delta > 0$, let us pick $X_i$ i.i.d. distributed as per a complex Gaussian distribution with mean $0$ and variance $P-\delta$. This will ensure with high probability that the codewords generated will satisfy the average power constraint for sufficiently large $n$. 

We rewrite, for $i>N$,
\begin{equation}
S_i =\frac{P-\delta}{2(2N+1)}\sum_{k=i-N}^{i+N} 2(X_{kR}^2 + X_{kI}^2)/(P-\delta),
\label{Sia}
\end{equation}
where $X_{kR}$ and $X_{kI}$ are the real and imaginary part respectively of $X_k$.
Ignoring the scaling coefficient $\frac{P-\delta}{2(2N+1)}$, we see that under our choice of input distribution, the sum on RHS has a standard central chi-squared distribution with  $4N+2$ degrees of freedom. We then have $\mathbb{E}[S_i^3] = \overline{S}(P-\delta)$, where
\begin{equation}
\overline{S}(P) \triangleq   \frac{P^3(2N+3)(2N+2)}{(2N+1)^2}.
\label{si3}
\end{equation}
Note that the RHS above does not depend on the index $i$.

We now state the main theorem of this paper which we will prove in the next section.
\begin{theo}\label{mainth}
The capacity $C$ of the channel described in (\ref{chan1}), under average power constraint $P$ is lower bounded by
\begin{equation}
 \log_2\left(1+\frac{\min\{P,P_N^*\}}{\sigma_A^2+\eta \overline{S}(\min\{P,P_N^*\})}\right),
\label{LB}
\end{equation}
where $\overline{S}(P)$ is as given in (\ref{si3}) and 
\begin{equation}
P_N^* = \left(\frac{\sigma_A^2(2N+1)^2}{2\eta(2N+3)(2N+2)}\right)^{1/3}. \hfill \qed
\label{popt}
\end{equation}
\end{theo}
A quick look at the bound reveals that it may be the capacity of an equivalent AWCGN channel of some sort. This is exactly the strategy we use to obtain the  lower bound. In the following, we further improve the bound (\ref{LB}) by considering a mixture of complex Gaussian inputs.

\subsection{Mixture of Complex Gaussian Inputs}
Given a positive integer $K$, non negative weights $\{\alpha_k\}$ that sum to $1$, $\mu_k \in \mathbb{C}$ and covariance matrices $\mathbf{P}_k$, we say that a random variable $X$ has a complex Gaussian mixture (CGM) distribution of order $K$ with the aforementioned parameters if
\begin{equation}
X \sim \sum_{k=1}^K \alpha_k \mathcal{CN}(\mu_k;\mathbf{P}_k),
\label{GMM1}
\end{equation}
where $\mathcal{CN}(\mu_k;\mathbf{P}_k)$ is the pdf. of a complex Gaussian random variable with mean $\mu_k$ and covariance matrix $\mathbf{P}_k$. Note that we will only consider \emph{proper} complex Gaussians i.e. complex Gaussian random variables with vanishing relation matrices and hence it suffices to specify only the mean and covariance matrix. 

There are several salient features of CGM distributions. 
\begin{enumerate}
	\item They are dense (in the sense of weak convergence) in the space of all proper distributions in $\mathbb{C}$ (Refer Appendix for a proof). Thus by studying the achievable rates obtained via CGM distributions, one can get good lower bounds on channel capacity.
	\item The circularly symmetric complex Gaussian, which was the input distribution used earlier, is a special case of a CGM with zero mean, $K=1$ and $\mathbf{P}_k=\frac{P}{2}I_{2\times 2}$. Thus by optimizing over the coefficients in the mixture, we can only improve our results.
	\item When a CGM is input to an AWCGN channel, the output is also a corresponding CGM. As our main strategy was to compare with an auxiliary channel, which was chosen to be AWCGN, this will help simplify the analysis.
 \end{enumerate}

Repeating the analysis in Section \ref{proo} but with i.i.d. CGM distribution with weights $\{\alpha_k\}$, mean $\{\mu_k\}$ and covariance matrices $\mathbf{P}_k=\frac{P_k}{2}I_{2\times 2}$, satisfying the power constraint, as input, we describe the following optimization problem.
\begin{equation}
\begin{aligned}
& \underset{\{\alpha_k,\mu_k,P_k\}}{\text{maximize}}
& & -\mathbb{E}\left[\log\left(\sum_{k=1}^K \alpha_k \frac{\exp\left\{ -\frac{|Y-\mu_k|^2}{(P_k + \sigma_A^2 +\eta \overline{S}^{(K)})}  \right\}}{\sqrt{\pi (P_k + \sigma_A^2 +\eta \overline{S}^{(K)})}}\right)\right] \\ 
& & & -\log\left(\pi e(\sigma_A^2 + \eta \overline{S}^{(K)})\right) \\
& \text{such that}
& &\sum_{k=1}^K \alpha_k =1,~ \sum_{k=1}^K \alpha_k \mu_k = 0, \\ 
& & & \sum_{k=1}^K \alpha_k(P_k + |\mu_k|^2) = P, ~ \alpha_k \geq 0, P_k \geq 0 ~\forall k.
\end{aligned} \label{optz}
\end{equation}
where $\overline{S}^{(K)}=\mathbb{E}[S_i^3]$ with CGM inputs. We discuss the results in Section \ref{numer}.

\section{Proof of Theorem \ref{mainth}} \label{proo}
We first introduce a lemma that appeared in papers by Arimoto \cite{arim} and Blahut \cite{blah}.
\begin{lem}[Auxiliary Channel Lower Bound (ACLB)] \label{topsl}
Let $W$ be the original channel and $\overline{W}$ any other channel with input and output alphabet of $W$. Then given any input distribution $P_X$, we have
\begin{equation}
I(P_X;W) \geq \mathbb{E}\left[\log\left( \frac{\overline{W}(Y|X)}{P_X\overline{W}(Y)}\right)\right],
\label{ACLB}
\end{equation}
where the expectation is with respect to $P_X(x)W(y|x)$. 
\end{lem}

A proof for the above lemma may be found in \cite{arn}. The advantage of this lemma is that the capacity of a ``difficult'' channel may be lower bounded by the capacity of a ``manageable'' channel. The accuracy of this approximation, depends on the choice of the auxiliary channel. 

Consider a new channel $\overline{W}$, where the output $\overline{Y}_i$ is given by 
\begin{equation}
\overline{Y}_i = X_i + A_i + Z_i\sqrt{\eta \overline{S}(P-\delta)},
\label{auxch}
\end{equation}
where $\overline{S}(P-\delta)$ is now a constant, given $P$, $\delta$ and $N$. Unlike (\ref{chan1}), this is a memoryless AWCGN channel and it is one of the simplest channels to analyze, making it a good choice for an auxiliary channel. 
Before we proceed, we will prove the following lemma.
\begin{lem}\label{varl}
Under $\{X_i\}$ i.i.d. complex Gaussian inputs as defined in Section \ref{Lb}, we have
\begin{enumerate}
	\item $\mathbb{E}[Y_i] = \mathbb{E}[\overline{Y}_i] = 0, $ \label{mean1}
	\item 	$\mathbb{E}[|Y_i|^2] = \mathbb{E}[|\overline{Y}_i|^2] = P-\delta+\sigma_A^2 +\eta  \overline{S}(P-\delta), $ \label{eqvar}
	\item $\mathbb{E}[|Y_i - X_i|^2] = \mathbb{E}[|\overline{Y}_i -X_i|^2] = \sigma_A^2 +\eta  \overline{S}(P-\delta). $ \label{eqvar2}
	\end{enumerate}
\end{lem}
\begin{proof}
We have
\begin{IEEEeqnarray}{rCl}
\mathbb{E}[Y_i] &=& \mathbb{E}\left[X_i + A_i + Z_i\sqrt{\eta S_i^3}\right]  \notag \\
&=&   \mathbb{E}[X_i] + \mathbb{E}[A_i] + \mathbb{E}[Z_i]\mathbb{E}\left[\sqrt{\eta S_i^3}\right] \label{ind1}\\
&=&0,
\label{meanp1}
\end{IEEEeqnarray}
where in (\ref{ind1}) we used the fact that $Z_i$ and $S_i$ are independent. Similarly, we can show that $\mathbb{E}[\overline{Y}_i] = 0$. 

For the second part, we observe that even though $X_i$ and $Z_i\sqrt{\eta S_i^3}$ are dependent variables, they are uncorrelated. They are both zero mean and 
\begin{equation}
\mathbb{E}[X_iZ_i^*\sqrt{\eta S_i^3}] = \mathbb{E}[Z_i^*]\mathbb{E}[X_iZ_i\sqrt{\eta S_i^3}]=0.
\label{uncorr1}
\end{equation}
Hence we have
\begin{IEEEeqnarray}{rCl}
\mathbb{E}[|Y_i|^2] = Var(Y_i) &=& \mathbb{E}[|X_i|^2] +  \mathbb{E}[|A_i|^2] + \eta \mathbb{E}[|Z_i|^2] \mathbb{E}[S_i^3] \notag\\
&=& P-\delta + \sigma_A^2 + \eta \overline{S}(P-\delta).
\label{varcal}
\end{IEEEeqnarray}
Similarly, we get the same value for $\mathbb{E}[|\overline{Y}_i|^2]$.

The proof of the third part is the same as the proof of the second part with $X_i=0$ in (\ref{varcal}).
\end{proof}

Now we consider $I(\mathbf{X}^{n+N};\mathbf{Y}^{n})$. Extending Lemma \ref{topsl} to vectors, we get
\begin{equation}
I(\mathbf{X}^{n+N};\mathbf{Y}^{n}) \geq \mathbb{E}\left[\frac{\overline{W}(\mathbf{Y}^n|\mathbf{X}^{n+N})}{P_{\mathbf{X}}\overline{W}(\mathbf{Y}^n)}\right],
\label{topvec}
\end{equation}
where the expectation is taken with respect to the joint density of the input distribution and the original channel density. Note that we have
\begin{equation}
\overline{W}(\mathbf{y}^n|\mathbf{x}^{n+N}) = \overline{W}(\mathbf{y}^n|\mathbf{x}^{n}) = \prod_{i=1}^n \overline{W}(y_i|x_i)
\label{dmc1}
\end{equation}
and
\begin{equation}
P_{\mathbf{X}}\overline{W}(\mathbf{y}^n) = \prod_{i=1}^n P_X\overline{W}(y_i),
\label{dmc2}
\end{equation}
due to the iid nature of inputs and because $\overline{W}$ is a memoryless channel. Substituting these values in (\ref{topvec}), yields
\begin{IEEEeqnarray}{rCl}
I(\mathbf{X}^{n+N};\mathbf{Y}^{n}) &\geq& n\log_2\left(1+\frac{P-\delta}{\sigma_A^2+\eta \overline{S}(P-\delta)}\right) \notag\\
&-& \mathbb{E}\left[ \frac{\|Y^n-X^n\|^2}{(\sigma_A^2 +\eta \overline{S}(P-\delta))}\right]\log_2(e) \label{deriv} \\
&+&  \mathbb{E}\left[ \frac{\|Y^n\|^2}{(P-\delta+\sigma_A^2 +\eta \overline{S}(P-\delta))}\right]\log_2(e).\notag
\end{IEEEeqnarray}
We have by Lemma \ref{varl},
\begin{IEEEeqnarray}{rCl}
\mathbb{E}[\|Y^n\|^2] &=& \sum_{k=1}^n \mathbb{E}[|Y_i|^2] \notag\\
&=& n(P-\delta+\sigma_A^2 +\eta  \overline{S}(P-\delta))
\label{linex1}
\end{IEEEeqnarray}
and 
\begin{equation}
\mathbb{E}[\|Y^n-X^n\|^2]  = n(\sigma_A^2 +\eta  \overline{S}(P-\delta)).
\label{linex2}
\end{equation}
Putting it all together and after taking limits as $n \to \infty$, we take $\delta \to 0$. Thus we get
\begin{equation}
C \geq \lim_{n \to \infty}\frac{I(\mathbf{X}^{n+N};\mathbf{Y}^{n})}{n} \geq \log_2\left(1+\frac{P}{\sigma_A^2+\eta \overline{S}(P)}\right).
\label{semif}
\end{equation}
The function on RHS, as a function of $P$, attains its maximum at $P_N^*$ given in (\ref{popt}). For $P<P_N^*$, it is a monotone increasing function of $P$. However after $P_N^*$, it decreases towards $0$. Given powers $P_1$ and $P_2$ such that $P_1<P_2$, we note that a codeword satisfying average power constraint $P_1$ for this channel, also satisfies the constraint for $P_2$. Hence an achievable rate for $P_1$ is also achievable under $P_2$. This means that for this channel, the capacity is monotone non-decreasing as a function of average power constraint $P$. 


With the above argument, we have finally proved (\ref{LB}).
\subsection{Optimality on choice of variance}
While using Lemma \ref{topsl}, the auxiliary channel was chosen by replacing $Z_i\sqrt{\eta S_i^3}$ with $Z_i\sqrt{\eta \overline{S}(P-\delta)}$. Now what if there is a constant $V$ such that replacing $\overline{S}(P-\delta)$ with $V$ leads to a better lower bound? It turns out that $\overline{S}(P-\delta)$ is the optimal choice for that constant. To see this, substitute $ \overline{S}(P-\delta)$ with $V$ in (\ref{deriv}). We get 
\begin{IEEEeqnarray}{rCl}
h(V) &=& n\log_2\left(1+\frac{P-\delta}{\sigma_A^2+\eta V}\right) \notag\\
&-& \mathbb{E}\left[ \frac{\|\mathbf{Y}^n-\mathbf{X}^n\|^2}{(\sigma_A^2 +\eta V)}\right]\log_2(e) \notag \\
&+&  \mathbb{E}\left[ \frac{\|\mathbf{Y}^n\|^2}{(P+\sigma_A^2 +\eta V)}\right]\log_2(e).
\label{optchk}
\end{IEEEeqnarray}
Using Lemma \ref{varl}, we get
\begin{IEEEeqnarray}{rCl}
h(V) &=& n\log_2\left(1+\frac{P-\delta}{\sigma_A^2+\eta V}\right) \notag\\
&-& \frac{\sigma_A^2 +\eta  \overline{S}(P-\delta)}{\sigma_A^2 +\eta V}\log_2(e) \notag \\
&+&  \frac{P-\delta+\sigma_A^2 +\eta \overline{S}(P-\delta) }{P-\delta+\sigma_A^2 +\eta V}\log_2(e).
\label{optchk2}
\end{IEEEeqnarray}
Differentiating with respect to $V$ and equating to $0$  gives us $V=\overline{S}(P-\delta)$ as the maximizer.
\subsection{Variation of  lower bound with memory} \label{memv}
Consider $\overline{S}(P)$ as defined in (\ref{si3}). For fixed $P$, $\overline{S}(P)$ decreases as memory $N$ increases, which causes an overall increase in the lower bound. Hence with increasing memory, our lower bounds increase for a fixed $P$. This is illustrated in Fig. \ref{plot1} for the example in Section \ref{numer}.

If we were to consider the scenario where i.i.d. inputs (not necessarily Gaussian but with $\mathbb{E}[X_i^2]=P$) are provided to the channel, then 
\begin{equation}
S_i = \frac{1}{2N+1}\sum_{k=i-N}^{i+N} |X_k|^2 \to P,\quad \mbox{a.s.}
\label{si4}
\end{equation}
as $N\to \infty$, by the strong law of large numbers \cite{athr}. 

Consider the channel
\begin{equation}
Y_i = X_i + A_i + Z_i\sqrt{\eta P^3},
\label{GNm}
\end{equation}
with average input power constraint $P$. This is nothing but the finite memory channel where $S_i$ is replaced with $P$ and is known as the Gaussian noise (GN) model for fiber optic channels (see \cite{Duris1}) . This channel is often considered (see \cite{pog1,pog2} and \cite{agr2}) while modeling non-linearities without considering the effects of memory. As this happens to be an AWGN channel, we may apply Shannon's channel capacity theorem \cite{cover} to obtain capacity as
\begin{equation}
C_{G}= \log_2\left(1+\frac{P}{\sigma_A^2 + \eta P^3}\right).
\label{GNC}
\end{equation}
We see that this capacity function first increases as a function of $P$, reaches its maximum value at
\begin{equation}
P_{G}^* = \left(\frac{\sigma_A^2}{2\eta} \right)^{1/3}
\label{GNpeak}
\end{equation}
 and then decreases to $0$. Recalling the definition of $P_N^*$ in (\ref{popt}), we see that 
\begin{equation}
\lim_{N\to \infty} P_N^* = P_{G}^*. 
\label{GNcon}
\end{equation}

Hence for $P \leq P_G^*$, we deduce that the lower bound tends to the capacity of the GN model with increasing memory. Fig. \ref{plot1} illustrates this effect for the example in Section \ref{numer}. 
\begin{figure}%
\centering
\includegraphics[scale=0.43]{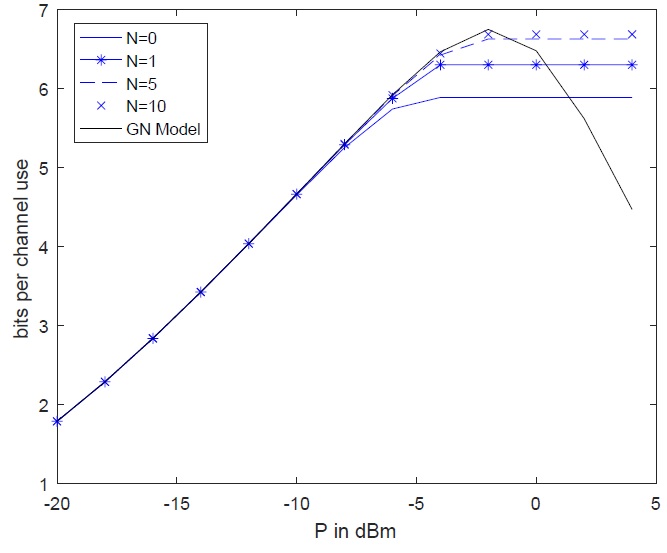}%
\caption{Plot of lower bounds with increasing $N$ in comparison with the capacity of GN model.}%
\label{plot1}%
\end{figure}

\section{Algorithm for channel capacity lower bound} \label{algorz}
We now describe an algorithm that calculates a lower bound on channel capacity for the channel with memory. The algorithm follows principles discussed in \cite{vs1,arn}. This could yield results better than our theoretical lower bound since we directly work with the given channel i.e. there is no auxiliary channel. Moreover we have the freedom to pick any input distribution albeit with quantization. For simplicity, we will assume that the memory $N=1$ throughout our discussion.

\subsection{Message passing algorithm for i.i.d. inputs}
To compute the mutual information, we need to compute $h(\mathbf{Y}^n)$ and $h(\mathbf{Y}^n|\mathbf{X}^{n+1})$. We note that 
\begin{equation}
h(\mathbf{Y}^n|\mathbf{X}^{n+1}) = \sum_{i=1}^n\mathbb{E}\left[\log\left(\pi e \left(\sigma_A^2 + \eta S_i^3\right)\right)\right].
\label{conen1}
\end{equation}
 However due to the i.i.d. nature of the inputs, all terms in the summation (except $i=1$) will be equal. Hence this expectation can be efficiently evaluated using Monte Carlo integration methods. We focus on computing $h(Y^n)$.

We suitably modify the message passing algorithm given in \cite{arn,vs1} to describe the channel in (\ref{chan1}). The algorithm essentially computes a Monte Carlo estimate of mutual information. We require that the input distribution have finite support and this can be achieved by quantizing continuous distributions.
\begin{algorithmic}[1]
\STATE{Generate $N_s+1$ samples of $\{x_k\}$, distributed according to $p_X$ that satisfies the input power constraints and having finite support. $N_s$ is the number of iterations we will run the algorithm to get satisfactory convergence.}
\STATE{Compute the corresponding output samples $\{y_k\}$ via (\ref{chan1}).}
\STATE{
\begin{IEEEeqnarray}{rCl}
\mu_1(x_1,x_2) &\leftarrow& p_X(x_1)p_X(x_2)p_{Y|X_1,X_2}(y_1|x_1,x_2) \notag\\
&&\forall (x_1,x_2).\notag
\label{tost}
\end{IEEEeqnarray}
}
\STATE{$\lambda_1 \leftarrow \sum_{x_1,x_2}\mu_1(x_1,x_2)$.}
\STATE{$\mu_1(x_1,x_2) \leftarrow \mu_1(x_1,x_2)/\lambda_1~ \forall (x_1,x_2)$. }
\FOR{$k=2$ to $N_s$}
\STATE{\begin{IEEEeqnarray}{lCl}
\mu_k(x_k,x_{k+1}) &\leftarrow& \sum_{x_{k-1}}\mu_{k-1}(x_{k-1},x_k)p_X(x_{k+1}) \notag\\
&\times&p_{Y_k|X_{k-1},X_{k},X_{k+1}}(y_k|x_{k-1},x_k,x_{k+1}) \notag \\
&\forall& (x_k,x_{k+1}). \notag
\label{rec1}
\end{IEEEeqnarray}
}
\STATE{$\lambda_k  \leftarrow \sum_{x_{k},x_{k+1}}\mu_k(x_{k},x_{k+1}) $.}
\STATE{$\mu_k(x_{k},x_{k+1}) \leftarrow \mu_k(x_{k},x_{k+1})/\lambda_k$.}
\ENDFOR
\STATE{$h(Y^n) = -\sum_{i=1}^{N_s} \log(\lambda_i)/N_s  $.} 
\end{algorithmic}
Analogous extensions are possible for $N > 1$. We use this algorithm to numerically evaluate a lower bound for channel capacity of (\ref{chan1}) in Section \ref{numer}.

\section{Numerical Results} \label{numer}
We compare our results with those provided in \cite{Duris1}. We use the same system parameters as in \cite{Duris1}. Hence we choose the nonlinearity parameter $\eta = 7244$ $\mbox{W}^{-2}$ and the ASE noise variance as $\sigma_A^2 = 4.1\times10^{-6}$ W. Fig. \ref{plot1} is a plot of our results using (\ref{LB}) and Fig. \ref{Durisgr} plots the lower bounds from \cite{Duris1}. 

\begin{figure}%
\includegraphics[scale=0.6]{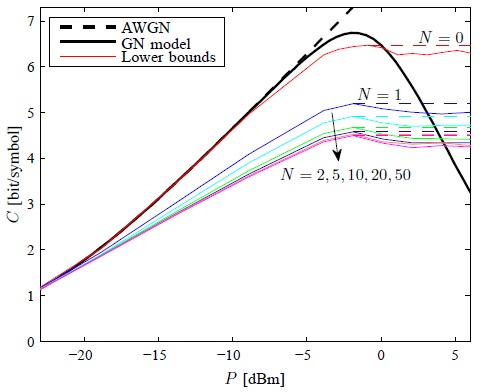}%
\caption{Plot of lower bounds of capacity from \cite{Duris1}.}%
\label{Durisgr}%
\end{figure}

For convenience, we have plotted these curves together in Fig. \ref{algc}. We infer from the plots that our bounds fare better than the bounds in \cite{Duris1} for $N\geq 1$. Moreover, unlike in \cite{Duris1}, our bounds increase with increasing memory $N$ as explained in Section \ref{memv}. 

We also evaluate the lower bound with i.i.d. quantized complex Gaussian inputs for varying input power constraints via the algorithm in Section \ref{algorz}. To do this, first generate a Gaussian random variable $g$ with variance $V$. Next, given $0 < \varepsilon < 1$, we split the interval $[-x_T,x_T]$ into $2(N_q-1)$ equal sub intervals, along with the intervals $[x_T,\infty)$ and $(-\infty,-x_T]$; where $x_T=2VQ^{-1}(\varepsilon/2)$, where $Q^{-1}$ is the inverse $Q$ function.Then if $g$ belongs to one of the subintervals, say $[a,b)$, then output $a$ if $g\geq 0$ or $b$ otherwise. It follows that the variance of the samples generated in this way will have variance less than $V$. To generate a complex quantized Gaussian sample $cg$, generate $g_1$ and $g_2$ quantized Gaussian with variance $V/2$ and use $cg=g_1+ig_2$. We take $N_q=20$ and $\varepsilon = 10^{-5}$ for quantization of the Gaussian input and $N_s=10000$. Fig. \ref{algc} compares the algorithmic bound (with monotone extension) with the theoretical results derived earlier. We see that the algorithmic bound improves over the theoretical bound in a few cases. The bound (\ref{LB}) is tighter than that in \cite{Duris1} and unlike the bound in \cite{Duris1}, is in closed form. We additionally ran the algorithm for a quantized uniform input but the best bound was given by the quantized Gaussian.

The optimization problem in (\ref{optz})  was solved using MATLAB's global optimization toolbox with the number of mixtures $K=5,10$ and $20$ and the results indicate that the optimum is within 1 percent of the theoretical result. This implies that, for our channel, the Gaussian lower bound (\ref{LB}) is good.
\begin{figure}%
\includegraphics[width=\columnwidth]{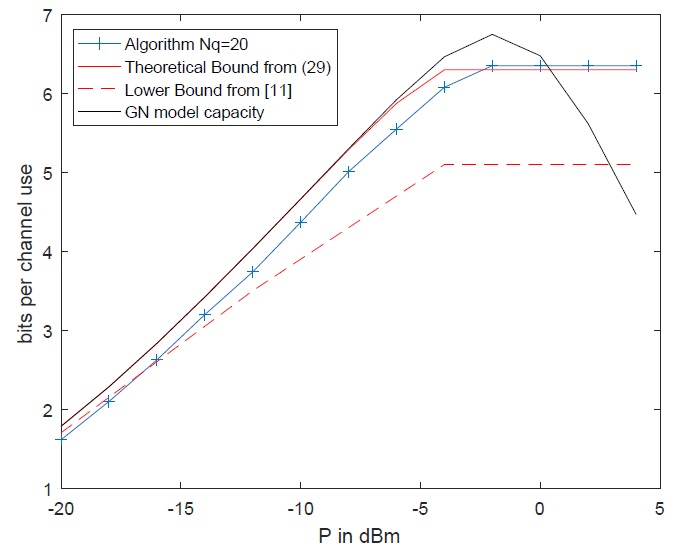}%
\caption{Comparison of algorithmically calculated lower bound with theoretical lower bound for N=1.}%
\label{algc}%
\end{figure}

\section{Conclusion} \label{concl}
We have derived a new improved lower bound in closed form on the channel capacity of an optical channel with memory. Unlike the previously available bound, we have shown that the bound improves with increase in memory. This bound could possibly be tightened by choosing a non Gaussian input and/or with a correlated input sequence to the channel. By obtaining a lower bound via a mixture of Gaussian distributions and via the general algorithm, the Gaussian bound seems to be tight. A useful, computable upper bound on the channel capacity will help in further checking the accuracy of the lower bounds. 


\section*{Acknowledgment}
The authors would like to thank Deekshith P.K. for his helpful comments and suggestions while preparing this manuscript.
\appendix[Denseness of CGM distributions]
The following proof is derived from Chapter 3, Theorem 4.2 of \cite{asmus} with some modifications. It suffices to look at distributions on $\mathbb{R}^2$ as any distribution on $\mathbb{C}$ is equivalent to a distribution on $\mathbb{R}^2$. Moreover, according to \cite{asmus}, it suffices to show that distributions with finite support over a compact set, say $[-A,A]^2$, are approximated in a weak sense by CGM distributions.

Let $\mathcal{P}_A$ be the set of all distributions with finite support over $[-A,A]^2$. Consider $F\in \mathcal{P}_A$ and suppose it has $K$ atoms $\{(x_k,y_k)\}_{k=1}^K$ with corresponding probabilities $\{\alpha_k\}_{k=1}^K$. Observe that the degenerate distribution of every atom $\{(x_k,y_k)\}$ is weakly approximated by the joint normals $\{\mathcal{N}(x_k,1/m)\times \mathcal{N}(y_k,1/m)\}_{m\geq 1}$.

Hence with
\begin{equation}
\sum_{k=1}^K \alpha_k  \{\mathcal{N}(x_k;1/m)\times \mathcal{N}(y_k;1/m)\} \stackrel{w}{\to} F,
\label{mixg1}
\end{equation}
and the fact that we may replace the corresponding joint normals above with $\mathcal{CN}(x_k+iy_k;\mathbf{P_m})$ where $\mathbf{P_m}=\bigl[\begin{smallmatrix}1/m & 0\\ 0 & 1/m\end{smallmatrix}\bigr]$, we get our result.

\bibliographystyle{IEEETran}
\bibliography{Globib}
%
%

\end{document}